\documentclass[11pt,reqno]{amsart}
\allowdisplaybreaks[4]
\usepackage{amsthm} 

\newtheorem{example}{Example}
\usepackage{graphicx} 
\usepackage{epsfig}
\usepackage{amsthm}
\usepackage{amssymb}
\usepackage{amsmath}
\usepackage{tcolorbox}
\usepackage{cite}
\newtheorem{theorem}{Theorem}

\newtheorem{proposition}{Proposition}
\newtheorem{corollary}{Corollary}

\newtheorem{lemma}{Lemma}
\newtheorem{remark}{Remark}
\newcommand{\be}{\begin{equation}}
	\newcommand{\ee}{\end{equation}}
\newcommand{\bea}{\begin{eqnarray}}
	\newcommand{\eea}{\end{eqnarray}}
\newcommand{\ba}{\begin{array}}
	\newcommand{\ea}{\end{array}}
\newcommand{\bean}{\begin{eqnarray*}}
	\newcommand{\eean}{\end{eqnarray*}}

\newcommand{\la}{\lambda}

\newcommand{\Om}{\Omega}

\newcommand{\pa}{\partial}

\begin{document}
	
	\title{The tau functions of the constrained CKP hierarchy}
	\author{Danqi Chen$^{1}$, Jipeng Cheng$^{1,2*}$, Shen Wang$^{1}$}
	\dedicatory {$^{1}$ School of Mathematics, China University of
		Mining and Technology, \ Xuzhou, Jiangsu 221116, China,\\
		$^{2}$ Jiangsu Center for Applied Mathematics (CUMT), \ Xuzhou, Jiangsu 221116, China.
	}

	\thanks{*Corresponding author. Email: chengjp@cumt.edu.cn;\ chengjipeng1983@163.com.}
	\begin{abstract}
	The CKP hierarchy is one important sub-hierarchy of the KP hierarchy, which is quite special due to its tau function. Here we construct the tau functions for the constrained CKP hierarchy $(L^k)_{<0}=\sum_{i=1}^{m}\big(q_{1,i}\partial^{-1}q_{2,i}-(-1)^kq_{2,i}\partial^{-1}q_{1,i}\big)$ with $k$ being odd or even positive integer by using the CKP Darboux transformations.
	\end{abstract}
	
	\maketitle

\section{Introduction}

The Kadomtsev-Petviashvili (KP) hierarchy is one of the important research topics in theoretical and mathematical physics \cite{date1983,dickey2003,Mulase1994}. The CKP hierarchy \cite{Arthamonov-2023, chang2013, date1981JPAJ, Krichever-CMP-2021, Loris-1999, vandeleur2012, WangS-Non-2025, yang2021}, a sub-hierarchy of the KP hierarchy corresponding to the infinite-dimensional Lie algebra $c_\infty$ \cite{date1981JPAJ, vandeleur2012, yang2021, Kacvandeleur2023}, is widely applied in 2D conformal field theories, topological field theories and matrix models (e.g. \cite{lishihao,lichunxia,Krichever-CMP-2021,Arthamonov-2023}), where its tau functions act as partition functions. The CKP hierarchy is defined by the following Lax equation:
	\begin{align}
		L_{t_n} = [(L^n)_{\geq0},L], \quad n=1,3,5,\ldots,\label{CKP Lax eq}
	\end{align}
	with the Lax operator $L=\partial+\sum_{i=1}^{\infty}u_{i+1}\partial^{-i}$ satisfying the CKP constraints below
	\begin{align}
		L^* = -L\label{CKP-con}.
	\end{align}
	Here $\partial=\partial_x$, $u_i=u_i(t)$, $t=(t_1=x, t_3,t_5,\ldots)$, the projection $(\sum_i a_i \partial^i)_{\geq0}=\sum_{i\geq0}a_i\partial^i$, and the adjoint operation * is defined by $(\sum_ia_i\partial^i)^*=\sum_i(-1)^i\partial^ia_i$. In what follows, $n$ is usually the positive odd number. 
	
	 Furthermore, the Lax operator $L$ can be described by the dressing operator $W=1+\sum_{i=1}^{\infty}w_i\partial^{-i}$ in the following way: $L=W \partial  W^{-1}$. Then the CKP hierarchy \eqref{CKP Lax eq} \eqref{CKP-con} can also be expressed by
	\begin{align}
		W_{t_n}=-(L^n)_{<0}W, \quad W^*=W^{-1}.\label{CKP-express}
	\end{align}
    If define the wave function $\psi(t,\lambda)=W(e^{\xi(t,\lambda)})$ with $\xi(t,\lambda)=\sum_{i=1}^{+\infty}t_{2i-1}\lambda^{2i-1}$, we can find $L(\psi)=\lambda \psi$, $\pa_{t_n}\psi=(L^n)_{\geq0}\psi$. The whole CKP hierarchy is equivalent to the following bilinear equation \cite{date1981JPAJ,Krichever-CMP-2021,Zabrodin2021, WangS-Non-2025}:
	\begin{align}\label{bili}
		{\rm Res}_\lambda \psi(t,\lambda)\psi(t',-\lambda)=0,
	\end{align}
	where ${\rm Res}_\lambda\sum_i a_i\lambda^i=a_{-1}$. In \cite{chang2013}, it is proved that there exists a tau function $\tau(t)$ such that
	\begin{align}\label{bili*}
		\psi(t,\lambda)
		=e^{\xi(t,\lambda)}\sqrt{\varphi(t,\lambda)}\frac{\tau(t-2[\lambda^{-1}])}{\tau(t)},\quad \varphi(t,\lambda)=1+\lambda^{-1}\partial_x\log\frac{\tau(t-2[\lambda^{-1}])}{\tau(t)},
	\end{align}
	where $[\lambda^{-1}]=(\lambda^{-1},\frac{1}{3}\lambda^{-3},\frac{1}{5}\lambda^{-5},\ldots)$. Notice that there exist one square root term in the relation \eqref{bili*} of the CKP wave function with the tau function, which brings much difficulty and challenge in the study of the CKP hierarchy. For the other integrable hierarchies \cite{Krichever-Zabrodin-LMP-2021, Zabrodin-2023, Wangshen2025} related with the CKP hierarchy, there are the similar properties in the tau functions. Therefore, the associated integrable hierarchies derived from the CKP hierarchy have also attracted considerable attention in recent years \cite{Krichever-Zabrodin-LMP-2021, yang2021, Zabrodin-2023,WangS-PhD-2026,Wangshen2025,Arthamonov-2023,Kacvandeleur2023}.

	In this paper, we will focus on the study of the following system:
	\begin{align}
		&(L^k)_{<0}=\sum_{i=1}^{m}\big(q_{1,i}\partial^{-1}q_{2,i}-(-1)^kq_{2,i}\partial^{-1}q_{1,i}\big),
		\label{Lax op}\\
		&L_{t_n}=[B_n,L], \quad L^*=-L,\quad B_n=(L^n)_{\geq0}, \label{LAX} \\
		&\partial_{t_n}q_{a,i}=(L^n)_{\geq0}(q_{a,i}), \quad a=1,2, \quad 1\leq i \leq m,\label{LAx}
	\end{align}
	which is called the $(k,m)$-constrained CKP hierarchy \cite{Loris-1999,he2007jmp, loris2001jpa}. Here $k$ and $m$ are all fixed positive integers. Note that the constrained CKP hierachy is a special case of the constrained KP hierarchy \cite{ChengY-1992,ChengY-1994,Sidorenko-1991,OeveleStrampp} $(L^k)_{<0}=\sum_{i=1}^mf_i\partial^{-1}g_i$, which contains many famous integrable hierarchies, including Ablowitz-Kaup-Newell-Segur hierarchy, Yajima-Oikawa hierarchy and so on. For the constrained CKP hierarchy, there exist many interesting results, in the aspects of the Darboux transformations, symmetries, bilinear equations and so on \cite{he2007jmp, Loris-1999,loris2001jpa, Tiankl2011,yang2021,cheng2014}. However the corresponding results are mainly done for the constrained CKP hierarchy when $k$ is odd, and there are also few results for the tau functions expressed by  \eqref{bili*}.

Here we would like to point out that the $(k,m)$-constrained CKP hierarchy with even $k$ is firstly given in \cite{loris2001jpa}, which is not the classical symmetry reduction.
In this paper, we will consider the cases of odd and even $k$ for the constrained CKP hierarchy in the framework of the tau functions. Our main results are given as follows.
	\begin{theorem}\label{k,m}
			The system of the $(k,m)$-constrained CKP hierarchy \eqref{Lax op}-\eqref{LAx} is equivalent to the following bilinear equations:
			\begin{align}
				&{\rm Res}_\la \la^k\psi(t,\la)\psi(t',-\la)=\sum^m_{i=1}(q_{1,i}(t)q_{2,i}(t')-(-1)^kq_{2,i}(t)q_{1,i}(t')), \label{theorem:bili1}\\
				&{\rm Res}_\la \psi(t,\la)\Om(q_{a,j}(t'),\psi(t',-\la))=-q_{a,j}(t), \quad a=1,2, \quad j=1,2,\ldots,m,
				\label{theorem:bili2}
			\end{align}
			where $\psi(t,\la)=(1+\sum_{i=1}^{+\infty} w_i\la^{-i})e^{\xi(t,\la)}$ is the CKP wave function, and $\Om$ is the squared eigenfunction potential \cite{cheng2014,oevel1993pa} defined by $\Om(f_1,f_2)_{t_n}={\rm Res}_{\pa}(\pa^{-1}f_2L^nf_1\pa^{-1})$ for $f_{i,t_n}=(L^n)_{\geq0}(f_i)$ and the CKP Lax operator $L$. Here $\psi(t,\la)$ is related with CKP Lax operator $L$ by $L=W\pa W^{-1}$ and $W=1+\sum_{i=1}^{\infty}w_i\pa^{-i}$, with $w_i$ in $\psi(t,\la)$.

Further, if assume that $\tau$ is the $tau$ function of the $(k,m)$-constrained CKP hierarchy and set $\rho_{a,i}=\tau q_{a,i}$, then \eqref{theorem:bili1} \eqref{theorem:bili2} can be expressed by
			\begin{align}
&{\rm Res}_\lambda \lambda^k\sqrt{\varphi(t,\lambda)\varphi(t',-\lambda)}\tau(t-2[\lambda^{-1}])\tau(t'+2[\lambda^{-1}])e^{\xi(t-t',\lambda)}\nonumber\\
				&=\sum_{i=1}^m\left(\rho_{1,i}(t)\rho_{2,i}(t')-(-1)^k\rho_{2,i}(t)\rho_{1,i}(t')\right),\label{res1}\\
&{\rm Res}_\lambda \la ^{-1} e^{\xi(t-t',\lambda)}\sqrt{\frac{\varphi(t,\lambda)}{\varphi(t',-\lambda)}}\tau(t-2[\lambda^{-1}])\nonumber\\
&\times \left(\tau(t')\rho_{a,i}(t'+2[\lambda^{-1}])+\rho_{a,i}(t')\tau(t'+2[\lambda^{-1}])\right)=2\rho_{a,i}(t)\tau^2(t').\label{res2}
			\end{align}

	\end{theorem}
\begin{remark}\label{remark:omega1}
Notice that $(L^*)^n=-L^n$, so we have $\Omega(f_1,f_2)=\Omega(f_2,f_1)$ (please refer to \cite{cheng2014}). And the integral constant in $\Omega(q_{a,j},\psi)$ is fixed by the way below:
\begin{align*}
\Omega(q_{a,j},\psi)=(\pa^{-1}\cdot q_{a,j}\cdot W)\cdot(e^{\xi(t,\la)}).
\end{align*}
Here we firstly compute $\pa^{-1}q_{a,j}W$ and then use $\pa^l(e^{\xi(t,z)})=z^le^{\xi(t,z)}$ to fix the integral constant in $\Omega(q_{a,j},\psi)$.
\end{remark}
	
\begin{theorem}\label{THE}
Given $(h_1,\cdots,h_M)=(f_1,f_1^{(k)},\cdots,f_1^{(kN_1)},\cdots,f_m,f_m^{(k)},\cdots,f_m^{(kN_m)},\varphi_1,\cdots,\varphi_P) $ with $f_{i}$ $(1\leq i\leq m)$ satisfying $f_{i_{t_n}}=f_i^{(n)}$ $(f_i^{(n)}\triangleq \pa_x^n(f_i))$ and  $\varphi_j=e^{\xi(t,\la_j)}(1\leq j\leq P)$, we have
\begin{align*}
&\tau^{\{M\}}=\sqrt{IW_{M,M}(h_1,\ldots,h_M;h_1,\ldots,h_M)} ,\\
&\rho_{1,j}^{\{M\}}=\frac{IW_{M,M+1}(h_1,\ldots,h_M;h_1,\ldots,h_M,f_j^{(kN_j+k)})}{\sqrt{IW_{M,M}(h_1,\ldots,h_M;h_1,\ldots,h_M)}},\\
&\rho_{2,j}^{\{M\}}=\frac{(-1)^{M+p(j)}IW_{M-1,M}(h_1,\ldots,\widehat{h_{p(j)}},\ldots,h_M;h_1,\ldots,h_M)}{\sqrt{IW_{M,M}(h_1,\ldots,h_M;h_1,\ldots,h_M)}},
\end{align*}
satisfy \eqref{res1} \eqref{res2} for the $(k,m)$-constrained CKP hierarchy, where $M=m+P+\sum^m_{l=1}N_l$, $p(j)=j+\sum_{l=1}^jN_l$, $1\leq j \leq m$ and $$(h_1,\ldots,\widehat{h_l},\ldots,h_M)=(h_1,\ldots,h_{l-1},h_{l+1},\ldots,h_M).$$
Here $IW_{N,M}$ is the generalized Wronskian determinant \cite{he2002} defined by
\[
IW_{N,M}(g_1,  \dots, g_N; f_1, \dots, f_M) = \begin{vmatrix}
	\Omega (f_1,g_1)  & \dots & \Omega( f_M ,g_1)  \\
	\vdots & \ddots & \vdots \\
	\Omega( f_1 ,g_N)  & \dots & \Omega( f_M, g_N)  \\
	f_1 & \dots & f_M \\
    f_1^{(1)} & \dots & f_M^{(1)} \\
	\vdots & \ddots & \vdots \\
	f_1^{(M-N-1)} & \dots & f_M^{(M-N-1)}
\end{vmatrix}.
\]

Further, the $q_{a,i}$ in the $(k,m)$-constrained CKP hierarchy can be given by
\begin{align}
q_{1,i}&=\frac{IW_{M,M+1}(h_1,\ldots,h_M;h_1,\ldots,h_M,f_j^{(kN_j+k)})}{IW_{M,M}(h_1,\ldots,h_M;h_1,\ldots,h_M)}\\
q_{2,i}&=\frac{(-1)^{M+p(j)}IW_{M-1,M}(h_1,\ldots,\widehat{h_{p(j)}},\ldots,h_M;h_1,\ldots,h_M)}{IW_{M,M}(h_1,\ldots,h_M;h_1,\ldots,h_M)}
\end{align}
\end{theorem}
\begin{remark}
Here we would like to explain how to compute $\Omega(f,g)$. When one of $f$ and $g$ is the sum of $e^{\xi(t,\la_i)}$, $\Omega(f,g)$ should be computed by the way in Remark \ref{remark:omega1}, that is $$\Omega(f(t),e^{\xi(t,\la_i)})=\Big(\partial^{-1}\cdot f(t)\Big)(e^{\xi(t,\la_i)}),$$  which is just the infinite integral with respect to $x$ without any integral constant, that is $\Omega(f,g)=\int fg dx $.
When $f$ and $g$ are the polynomials of $t$, $\Omega(f,g)$ should be computed by
\begin{align*}
\Omega(f(t),g(t))=\sum_{i=0}^{+\infty}\int_0^1 t_iA_{2i+1}(yt) dy,
\end{align*}
where $A_{2i+1}(t)=\Omega(f,g)_{t_{2i+1}}={\rm Res}_{\pa} (\pa^{-1}g(L^{2i+1})_{\geq0}f\pa^{-1})$.
\end{remark}
The remaining part of this paper is as follows. The bilinear equations of the constrained CKP hierarchy are investigated in Section 2. In Section 3, the Darboux transformations of the constrained CKP hierarchy are discussed.\\

\section{The bilinear equations of the constrained CKP hierarchy}
In this section, we investigate some basic facts on the $(k,m)$-constrained CKP hierarchy. We first demonstrate that the system is well-defined, and then prove the equivalence between its Lax representation and bilinear equation, that is, Theorem 1.

\subsection{The $(k,m)$-constrained CKP hierarchy}
The $(k,m)$-constrained CKP hierarchy \cite{Loris-1999,loris2001jpa, he2007jmp} is defined as follows:
\begin{align}
	&(L^k)_{<0}=\sum_{i=1}^{m}\big(q_{1,i}\partial^{-1}q_{2,i}-(-1)^kq_{2,i}\partial^{-1}q_{1,i}\big),
	\label{Lax op2}\\
	&L_{t_n}=[B_n,L], \quad L^*=-L,\quad B_n=(L^n)_{\geq0}, \label{LAX2} \\
	&\partial_{t_n}q_{a,i}=B_n(q_{a,i}), \quad a=1,2, \quad 1\leq i\leq m\label{LAx2}.
\end{align}

Let us start from showing that the $(k,m)$-constrained CKP hierarchy is well-defined, which means that
 \begin{align*}
 ((L^k)_{<0})^*=(-1)^k(L^k)_{<0} ,\quad
 \partial_{t_n}((L^k)_{<0}-A_{k,m})=0,
\end{align*}
still hold after substituting \eqref{Lax op2}-\eqref{LAX2}. Here we assume $A_{k,m}\triangleq\sum_{i=1}^{m}\big(q_{1,i}\partial^{-1}q_{2,i}-(-1)^kq_{2,i}\partial^{-1}q_{1,i}\big)$. Notice that
\begin{align}
	A^*_{k,m}=(-1)^kA_{k,m},
\end{align}
 which implies $((L^k)_{<0})^*=(-1)^k(L^k)_{<0}$ by $(L^k)_{<0}=A_{k,m}$. We can know by $L_{t_n}=[B_n,L]$ and $(L^k)_{<0}=A_{k,m}$ that
\begin{align*}
	\partial_{t_n}(L^k)_{<0}=[B_n,(L^k)_{<0}]_{<0}=[B_n,A_{k,m}]_{<0}.
\end{align*}
Then by the formulas(please refer to \cite{oevel1993pa}) for a pseudo-differential operator $A$ and a function $f$:
\begin{align} (A_{\geq0}f\partial^{-1})_{<0}=A_{\geq0}(f)\cdot\partial^{-1}, (\partial^{-1}fA_{\geq0})_{<0}=\partial^{-1}\cdot A_{\geq0}^*(f),
\end{align}
we can obtain
\begin{align}\label{AKM}
	[B_n,A_{k,m}]_{<0}=\sum_{i=1}^m\big(B_n(q_{1,i})\cdot\partial^{-1}q_{2,i}-(-1)^kB_n(q_{2,i})\cdot\partial^{-1}q_{1,i}\nonumber\\
	+q_{1,i}\partial^{-1}\cdot B_n^*(q_{2,i})+(-1)^kq_{2,i}\partial^{-1}\cdot B_n^*(q_{1,i})\big).
\end{align}
Furthermore, with the help of $\partial_{t_n}q_{a,i}=B_n(q_{a,i})$, we can obtain that $\partial_{t_n}((L^k)_{<0}-A_{k,m})=0$ holds. Therefore the $(k,m)$-constrained CKP hierarchy is well-defined.																																																																																																
\begin{remark}
Here the $(k,m)$-constrained CKP hierarchy is the generation of the Gelfand-Dickey hierarchy of C type \cite{Loris-1999,loris2001jpa}, that is $(L^k)_{<0}=0$. For example, the Lax operator of the $(3,1)$-constrained CKP hierarchy is given by
$$L^3=\pa^3+3u\pa+\frac{3}{2}u_x+q_1\pa^{-1}q_2+q_2\pa^{-1}q_1.$$
It can be found that
\begin{align} u_{t_5}=&\frac{20}{3}q_1q_2u_x-\frac{15}{3}u^2u_x+\frac{20}{3}q_2q_{1,x}u+\frac{20}{3}q_{2,x}q_1u-\frac{5}{3}uu_{xxx}-\frac{25}{6}u_xu_{xx}\nonumber\\ &+\frac{20}{9}q_2q_{1,xxx}+\frac{20}{9}q_1q_{2,xxx}+\frac{10}{3}q_{2,x}q_{1,xx}+\frac{10}{3}q_{1,x}q_{2,xx}-\frac{1}{9}u_{xxxxx},\label{ut5}\\
q_{1,t_5}=&q_{1,xxxxx}+5uq_{1,xxx}+\frac{15}{2}u_xq_{1,xx}+(\frac{10}{3}q_1q_2+5u^2+\frac{35}{6}u_{xx})q_{1,x}\nonumber\\
	&+5uu_x+\frac{5}{3}q_1q_{2,x}+\frac{5}{3}q_{1,x}q_{2}+\frac{5}{3}u_{xxx},\nonumber \\ q_{2,t_5}=&q_{2,xxxxx}+5uq_{2,xxx}+\frac{15}{2}u_xq_{2,xx}+(\frac{10}{3}q_1q_2+5u^2+\frac{35}{6}u_{xx})q_{2,x}\nonumber\\
	&+5uu_x+\frac{5}{3}q_1q_{2,x}+\frac{5}{3}q_{1,x}q_{2}+\frac{5}{3}u_{xxx},\nonumber
	\end{align}
Notice that $q_1=q_2=0$, we can find that \eqref{ut5} becomes the Kaup-Kuperschmidt equation \cite{Konopelchenko1984}.
\end{remark}

\subsection{Proof of Theorem 1}
	First, we prove that \eqref{Lax op}-\eqref{LAx} can imply \eqref{theorem:bili1}-\eqref{theorem:bili2}. Note that \eqref{theorem:bili2} is the spectral representation of the  CKP eigenfunction, which can be found in \cite[Proposition 3]{cheng2014}. Therefore, we only need to prove \eqref{theorem:bili1} holds.
	Actually by
	$L^k=(L^k)_{\geq0}+A_{k,m}$, $L(\psi)=\la\psi$ and $\psi_{t_k}=B_k(\psi)$, we can know
	\begin{align*}
		\la^k\psi(t,\la)
		&=L^k(\psi(t,\la))=(L^k)_{\geq0}(\psi(t,\lambda))+A_{k,m}(\psi(t,\lambda))\\
		&=(L^k)_{\geq0}(\psi(t,\la))+\sum_{i=1}^m(q_{1,i}\Omega(q_{2,i},\psi)-(-1)^kq_{2,i}\Omega(q_{1,i},\psi)),
	\end{align*}
	implying
	\begin{align}\label{psi}
		(L^k)_{\geq0}(\psi)=\la^k\psi-\sum_{i=1}^m(q_{1,i}\Omega(q_{2,i},\psi)-(-1)^kq_{2,i}\Omega(q_{1,i},\psi)).
	\end{align}
	In addition, by applying $(L^k)_{\geq0}$ to the CKP bilinear equation ${\rm Res}_\lambda \psi(t,\lambda)\psi(t',-\lambda)=0$, we have
	\begin{align}\label{rest}
		{\rm Res}_\lambda (L^k)_{\geq0}(\psi(t,\lambda))\cdot\psi(t',-\lambda)=0.
	\end{align}
	Substituting the result for $(L^k)_{\geq0}(\psi)$ (see \eqref{psi}) into \eqref{rest}, and combining with \eqref{theorem:bili2}, we finally conclude that
\begin{align*}
{\rm Res}_\la \la^k\psi(t,\la)\psi(t',-\la)=\sum^m_{i=1}(q_{1,i}(t)q_{2,i}(t')-(-1)^kq_{2,i}(t)q_{1,i}(t')).  \end{align*}
	
	Conversely from \eqref{theorem:bili1}-\eqref{theorem:bili2} to \eqref{Lax op}-\eqref{LAx}, we need the lemmas below.
	\begin{lemma}\label{A}
		\cite{date1983,Cheng-Milanov} Set $A(x)=\sum_i a_i(x)\pa_x^i$, $B(x')=\sum_j b_j(x')\pa_{x'}^j$, then
		\begin{align}
			(A(x)B^*(x)\pa_x)(\Delta^0)={\rm Res}_{\la}A(x)(e^{x\la})B(x')(e^{-x'\la}),
		\end{align}
		where $\Delta^0=(x-x')^0$ and
		\begin{align}\label{pa}
			\begin{cases}
				\pa_x^{-a}(\Delta^0)=0, & a<0;\\
				\pa_x^{-a}(\Delta^0)=\dfrac{(x-x')^a}{a!}, & a\geq0.
			\end{cases}
		\end{align}
	\end{lemma}
\begin{lemma}\label{lemf}
\cite{Wuyq2023} For any two functions \(f(x)\) and \(g(x)\), the product \(f(x)g(x')\) can be expressed as
\[
f(x)g(x') = \left(f(x)\partial_x^{-1}g(x)\partial_x\right)\big(\Delta^0\big).
\]
\end{lemma}
First, by differentiating both sides of \eqref{theorem:bili2} with respect to $x'$, we obtain \begin{align*}
		{\rm Res}_\lambda \psi(t,\lambda)\psi(t',-\lambda)=0,
	\end{align*}
which is exactly the CKP bilinear equation. Note that $\psi(t,\la)=(1+\sum_{i=1}^{\infty}w_i\la^{-i})e^{\xi(t,\la)}$, and if we set $W=1+\sum_{i=1}^{\infty}w_i\pa^{-i}$, then $\psi(t,\la)=W(e^{\xi(t,\la)})$, and the following relations hold (please refer to \cite{dickey2003,Zabrodin2021})
\begin{align}
W_{t_n}=-(L^n)_{<0}W, \quad W^*=W^{-1}.\label{WL}
\end{align}
If set $L=W\pa W^{-1}$, then \eqref{LAX} holds. And by \eqref{WL}, we can get
\begin{align}
L(\psi)=\la \psi,\quad \pa_{t_n}\psi=(L^n)_{\geq0}(\psi).
\end{align}
Thus by $\pa_{t_n}\psi=(L^n)_{\geq0}(\psi)$ and \eqref{theorem:bili2}, we can know  $\pa_{t_n}q_{a,i}=(L^n)_{\geq0}(q_{a,i})$ holds. Therefore we only need to prove \eqref{Lax op}.

If we set $t'=t$ in \eqref{theorem:bili1} except $x'\neq x$, we can get by using Lemma \ref{A} and Lemma \ref{lemf} that
	\begin{align*}
		(W(t)\partial^kW(t)^*\pa)(\Delta^0)=\sum_{i=1}^m\left(q_{1i}(t)\pa^{-1}q_{2i}(t)\pa-(-1)^kq_{2i}(t)\pa^{-1}q_{1i}(t)\pa\right)(\Delta^0).
	\end{align*}
	Therefore, by \eqref{pa}
	\begin{align*}
		(W\partial^kW^*\pa)_{\leq0}=\sum_{i=1}^m\left(q_{1i}\pa^{-1}q_{2i}\pa-(-1)^kq_{2i}\pa^{-1}q_{1i}\pa\right).
	\end{align*}
So by $W^*=W^{-1}$, we can know
\begin{align*}
(W\partial^kW^{-1})_{<0}=\sum_{i=1}^m \left(q_{1i}\pa^{-1}q_{2i}-(-1)^kq_{2i}\pa^{-1}q_{1i}\right),
\end{align*}
	that is
\begin{align*}
(L^k)_{<0}=\sum_{i=1}^m \left(q_{1i}\pa^{-1}q_{2i}-(-1)^kq_{2i}\pa^{-1}q_{1i}\right),
\end{align*}
which is just \eqref{Lax op}.
Furthermore, $\Omega(q,\psi)$ is given by (refer to \cite{cheng2014,WangS-Non-2025})
\begin{align}
\Omega(q(t),\psi(t,\la))=\frac{1}{2\la\varphi(t,\la)}(q(t-2[\la^{-1}])+q(t))\psi(t,\la).
\end{align}
Therefore, if set $\rho_{a,i}=\tau q_{a,i}$, then we can obtain \eqref{res1}\eqref{res2}.
\section{The Darboux transformations of the constrained CKP hierarchy}
In this section, we firstly investigate the CKP Darboux transformation by the results of the KP Darboux transformation. And then based upon this, we give the solutions of the $(k,m)$-constrained CKP hierarchy, that is Theroem 2.

\subsection{The CKP Darboux transformations}
	We firstly review some basic facts about the CKP Darboux transformations.
A pseudo-differential operator $T$ is called the CKP Darboux operator \cite{he2007jmp,WangS-Non-2025}, if for a CKP Lax operator $L$, the operator
\begin{align}
L^{\{1\}}=TLT^{-1},
\end{align}
is still a CKP Lax operator. It is proved that the CKP Darboux operator $T$ satisfying
 \begin{align}
 T^{*-1}=T,
 \end{align}
 and the fundamental CKP Darboux operator is given by \cite{he2007jmp,WangS-Non-2025}
\begin{align}
 T[q]=1-\Omega(q,q)^{-1}q\partial^{-1}q,
\end{align}
   where $q$ is the CKP eigenfunction satisfying $q_{t_n}=(L^n)_{\geq0}(q)$ with CKP Lax operator $L$. And $T[q]$ can be constructed by the KP Darboux operators \cite{oevel1993pa,he2002,yang2022} $T_D[f]=f\partial f^{-1}, T_I[g]=g^{-1}\partial^{-1}g$ in the way below
\begin{align}
T[q]=T_D[q^{[1]}]T_I[q],
\end{align}
where $q^{[1]}=q^{-1}\Omega(q,q)$, $f$ is the KP eigenfunction and $g$ is the KP adjoint eigenfunction satisfying
\begin{align*}
f_{t_p}=(\mathcal{L}^p)_{\geq0}(f),\quad
g_{t_p}=-(\mathcal{L}^p)^*_{\geq0}(g),
\end{align*}
for the KP Lax operator $\mathcal{L}=\pa+\sum_{i=1}^{\infty}v_{i+1}\pa^{-1}$.

Under the action of the CKP Darboux operator $T[q]$, the CKP wave function $\psi(t,\la)$, and the CKP tau function $\tau(t)$ will be transformed in the way below \cite{he2007jmp,WangS-Non-2025}
\begin{align}
&\psi(t,\la)\to \psi^{\{1\}}(t,\la)=T[q](\psi(t,\la)),\\
&\tau(t)\to\tau^{\{1\}}(t)= \Omega (q(t),q(t))^{1/2}\tau(t),
\end{align}
Here $\psi^{\{1\}}$ and $\tau^{\{1\}}$ will still satisfy the CKP bilinear equation \eqref{bili} with \eqref{bili*} (please refer to \cite[Theorem 1]{WangS-Non-2025}).

Notice that the fundamental KP Darboux transformations \( T_d \) and \( T_i \) can commute with each other in the way below  \cite{oevel1993pa}
\begin{align*}
T_I[g^{[1]}]T_D[f] = T_D[f^{[1]}]T_I[g],\quad
T_D[f_1^{[1]}]T_D[f_2] = T_D[f_2^{[1]}]T_D[f_1],\quad
T_I[g_1^{[1]}]T_I[g_2] = T_I[g_2^{[1]}]T_I[g_1],
\end{align*}
where  $f_i^{[1]}=T_D[f_{3-i}](f_i)$, $g_i^{[1]}=(T_D[g_{3-i}]^*)^{-1}(g_i)$, $g^{[1]}=(T_D^{-1}[f])^*(g)$, $f^{[1]}=T_D[g](f)$. So we can only consider the following multi-step KP Darboux transformations

\[
\mathcal{L} \xrightarrow{T_D[f_1]} \mathcal{L}^{[1]} \xrightarrow{T_D[f_2^{[1]}]} \mathcal{L}^{[2]} \to \dots \to \mathcal{L}^{[n-1]} \xrightarrow{T_D[f_n^{[n-1]}]} \mathcal{L}^{[n]}
\]
\[
 \xrightarrow{T_I[g_1^{[n]}]} \mathcal{L}^{[n+1]} \xrightarrow{T_I[g_2^{[n+1]}]} \dots \to \mathcal{L}^{[n+k-1]} \xrightarrow{T_I[g_k^{[n+k-1]}]} \mathcal{L}^{[n+k]},
\]
where $\mathcal{L} $ is the KP Lax operator, $f_i \ (i=1,\dots,M)$ are the KP independent eigenfunctions and $g_j \ (j=1,\dots,N)$ are the KP adjoint eigenfunctions. Here we denote
\begin{align*}
	&L^{[i]} = T^{[i,0]} L T^{[i,0]-1},  \quad
	f_i^{[i-1]} = T^{[i,0]} (f_i),\quad 1 \leq i \leq M \\
	&L^{[M+j]} = T^{[M,j]} L T^{[M,j]-1}, \quad
	g_j^{[M+j-1]} = T^{[M,j]-1*} (g_j),\quad 1 \leq j \leq N,
\end{align*}
and
\begin{align*}
T^{[\overrightarrow{M},\overrightarrow{N}]}&=T^{[\overrightarrow{M},\overrightarrow{N}]}(f_1,\cdots,f_M;g_1,\cdots,g_N)\\
&= T_i[g_N^{[M+N-1]}] \dots T_i[g_1^{[M]}] T_d[f_M^{[M-1]}] \dots T_d[f_1],
\end{align*}
where $\overrightarrow{M}=(M,M-1,\cdots,2,1)$, $\overrightarrow{N}=(N,N-1,\cdots,2,1)$. Under the action of $T^{[\overrightarrow{M},\overrightarrow{N}]}$, the KP eigenfunction $\Phi$, the KP adjoint eigenfunction $\Psi$ and the KP tau function $\tau_{\rm KP}$ will be transformed in the way below \cite{he2002,yang2022}.
\begin{itemize}
\item When  $M > N$,
\begin{align*}
&\Phi^{[M+N]} = \frac{IW_{N,M+1}\bigl(g_N,\dots,g_1;\,f_1,\dots,f_M,\Phi\bigr)}{IW_{N,M}\bigl(g_N,\dots,g_1;\,f_1,\dots,f_M\bigr)} , \\
&\Psi^{[M+N]} = (-1)^M \frac{IW_{N+1,M}\bigl(\Psi,g_N,\dots,g_1;\,f_1,\dots,f_M\bigr)}{IW_{N,M}\bigl(g_N,\dots,g_1;\,f_1,\dots,f_M\bigr)} , \\
&\tau_{\rm KP}^{[M+N]} = (-1)^{MN} IW_{N,M}\bigl(g_N,\dots,g_1;\,f_1,\dots,f_M\bigr) \tau_{\rm KP}.
\end{align*}
\item When  $M = N $,
\begin{align}\label{2M=}
&\Phi^{[2M]} = \frac{IW_{M,M+1}\bigl(g_M,\dots,g_1;\,f_1,\dots,f_M,\Phi\bigr)}{IW_{M,M}\bigl(g_M,\dots,g_1;\,f_1,\dots,f_M\bigr)} ,\\
&\Psi^{[2M]} = \frac{(-1)^M \, IW_{M,M+1}\bigl(f_1,\dots,f_M;\,\Psi,g_M,\dots,g_1\bigr)}{IW_{M,M}\bigl(g_M,\dots,g_1;\,f_1,\dots,f_M\bigr)},\nonumber\\
&\tau_{\rm KP}^{[2M]} = (-1)^{M} \, IW_{M,M}\bigl(g_M,\dots,g_1;\,f_1,\dots,f_M\bigr)\tau_{\rm KP}.\nonumber
\end{align}
\item When  $M < N$,
\begin{align}\label{2M<}
&\Phi^{[M+N]} = \frac{(-1)^M \, IW_{M+1,N}\bigl(\Phi,f_M,\dots,f_1;\,g_1,\dots,g_N\bigr)}{IW_{M,N}\bigl(f_M,\dots,f_1;\,g_1,\dots,g_N\bigr)},\\
&\Psi^{[M+N]} = \frac{(-1)^{M+N} \, IW_{M,N+1}\bigl(f_M,\dots,f_1;\,g_1,\dots,g_N,\Psi\bigr)}{IW_{M,N}\bigl(f_M,\dots,f_1;\,g_1,\dots,g_N\bigr)},\nonumber\\
&\tau_{\rm KP}^{[N+M]} = (-1)^{MN+\frac{M(M-1)}{2}+\frac{N(N-1)}{2}} \, IW_{M,N}\bigl(f_1,\dots,f_M;\,g_N,\dots,g_1\bigr)\tau_{\rm KP}\nonumber.
\end{align}
\end{itemize}

By direct computations, we can have the following lemma
\begin{lemma}\label{lem:T}
For the KP Darboux operators $T_D$ and $T_I$, the KP eigenfunctions $f,f_1,f_2$ and the KP adjoint eigenfunctions $g,g_1,g_2$
	\begin{align}	
    T_D[f^{[1]}](g^{-1})&=(T_D^{-1}[f]^*(g))^{-1},\quad\quad
    f^{[1]}=g^{-1}\Omega(f,g) ,\label{td-1}\\
		(T_I^{-1}[g^{[1]}])^*(f^{-1})&=(T_I[g](f))^{-1},\quad\quad\ \ \ g^{[1]}=-f^{-1}\Omega(f,g),\label{ti-1}\\
		(T_D^{-1}[f_2^{[1]}])^*(f_1^{-1})&=(T_D[f_2](f_1))^{-1},\quad \quad
f_2^{[1]}=f_1 (f_1^{-1}f_2)_x,\label{td-2}\\
		T_I[g_2^{[1]}](g_1^{-1})&=(T_I^{-1}[g_2]^*(g_1))^{-1},\quad\  g_2^{[1]}=-g_1 (g_1^{-1}g_2)_x.\label{ti-2}
	\end{align}
\end{lemma}

If denote for the CKP eigenfunctions $f$ and $f_1,\cdots,f_M$,
\begin{align}\label{tm}
&T^{\{\overrightarrow{M}\}}(f_1,f_2,\cdots,f_M)\triangleq T^{[\overrightarrow{M},\overrightarrow{M}]}(f_1,\cdots,f_M;f_1,\cdots,f_M),\nonumber\\
&L^{\{M\}}=T^{\{\overrightarrow{M}\}}LT^{\{\overrightarrow{M}\}-1},\quad f^{\{M\}}=T^{\{\overrightarrow{M}\}}(f),
\end{align}
$T^{\{M\}}$ means $T^{\{M\}}(f_1,\cdots,f_M)$, $L^{\{M\}}$ is the new CKP Lax operator, and $f^{\{M\}}$ is the new CKP eigenfunction. Then under the action of the CKP Darboux operator $T^{\{\overrightarrow{M}\}}$, the CKP wave function $\psi(t,\la)$, and the CKP tau function $\tau(t)$ will be transformed in the way below (please refer to \cite [Corollary 6]{WangS-Non-2025})
\begin{align}
&\psi(t,\la)\to \psi^{\{M\}}(t,\la)=T^{\{\overrightarrow{M}\}}(\psi(t,\la)),\\
&\tau(t)\to\tau^{\{M\}}(t)=\sqrt{IW_{M,M}(f_1,\cdots,f_M;f_1,\cdots,f_M)}\tau(t).\label{mm}
\end{align}
Here $\psi^{\{M\}}$ still satisfies the CKP bilinear equation \eqref{bili} and is related with $\tau^{\{M\}}$ by \eqref{bili*}.
\subsection{The constrained CKP hierarchy under the CKP Darboux transformation} 
In this section, we will discuss the $(k,m)$-constrained CKP hierarchy under the action of $T^{\{\overrightarrow{M}\}}$. For this, let us see the lemma below.

\begin{lemma}\label{lemma:TLKT}
For the Lax operator L of the $(k,m)$-constrained CKP hierarchy:
\begin{align*}
(T[f]L^kT[f]^{-1})_{<0}
&=\left(T[f]L^k\right)(f)\cdot\pa^{-1}\cdot(T_I[f](f))^{-1}
-(-1)^k(T_I[f](f))^{-1}\cdot\pa^{-1}\cdot T[f]L^k(f)\\
&+\sum_{i=1}^m\big(T[f](q_{1i})\cdot\pa^{-1}\cdot T[f](q_{2i})-(-1)^kT[f](q_{2i})\cdot\pa^{-1}\cdot T[f](q_{1i}\big),
\end{align*}
where $f$ is the CKP eigenfunction.
\end{lemma}
\begin{proof}
Since $T[f]=T_D[f^{[1]}]T_I[f]$, let us first compute $\left(T_I[f]L^kT_I^{-1}[f]\right)_{<0}$ and then we can find (please refer to \cite{ChauShawTu1997})
\begin{align*}
&\left(T_I[f]L^kT_I^{-1}[f]\right)_{<0}
=f^{-1}\pa^{-1}\left(T_I[f]^{-1*}L^{*k}\right)(f)\\
+&\sum_{i=1}^m\left(T_I[f](q_{1i})\cdot\pa^{-1}\cdot T_I[f]^{-1*}(q_{2i})-(-1)^kT_I[f](q_{2i})\cdot\pa^{-1}\cdot T_I[f]^{-1*}(q_{1i})\right).
\end{align*}
Next we can get
\begin{align*}
\left(T_D[f^{[1]}]T_I[f]L^k\left(T_D[f^{[1]}]T_I[f]\right)^{-1}\right)_{<0}
=T[f]L^k(f)\pa^{-1}f^{[1]-1}
+T_D[f^{[1]}](f^{-1})\pa^{-1}T[f]^{-1*}L^{*k}(f)\\
+\sum_{i=1}^m\left(T[f](q_{1i})\pa^{-1}T[f]^{-1*}(q_{2i})-(-1)^kT[f](q_{2i})\pa^{-1}T[f]^{-1*}(q_{1i})\right).
\end{align*}
Owing to $T[f]^{*-1}=T[f]$ (the property of the CKP Darboux operator) and $L^*=-L$, together with \eqref{td-1}, we can get
\begin{align*}
(T[f]L^kT[f]^{-1})_{<0}
&=\left(T[f]L^k\right)(f)\cdot\pa^{-1}\cdot(T_I[f](f))^{-1}
-(-1)^k(T_I[f](f))^{-1}\cdot\pa^{-1}\cdot T[f]L^k(f)\\
&+\sum_{i=1}^m\left(T[f](q_{1i})\cdot\pa^{-1}\cdot T[f](q_{2i})-(-1)^kT[f](q_{2i})\cdot\pa^{-1}\cdot T[f](q_{1i}\right).
\end{align*}
\end{proof}
The more general case can be found in the proposition below.
\begin{proposition}\label{Lax_(k,m)}
For the Lax operator L of the $(k,m)$-constrained CKP hierarchy,
\begin{align*}
\bigl(T^{\{\overrightarrow{M}\}}L^kT^{\{\overrightarrow{M}\}-1} \bigr)_{<0}
&=\sum_{j=1}^{M}\bigl(T^{\{\overrightarrow{M}\}}L^k\bigr)(f_j)\cdot \partial^{-1}\cdot\bigl(T^{[\overrightarrow{M}\setminus\{j\},\overrightarrow{M}]}(f_j)\bigr)^{-1}\\
&-(-1)^k\sum_{j=1}^M\bigl(T^{[\overrightarrow{M}\setminus\{j\},\overrightarrow{M}]}(f_j)\bigr)^{-1}\cdot\partial^{-1}\cdot\bigl(T^{\{\overrightarrow{M}\}}L^k\bigr)(f_j)\\
&+\sum_{i=1}^{m}\big(T^{\{\overrightarrow{M}\}}(q_{1i})\cdot\partial^{-1}\cdot T^{\{\overrightarrow{M}\}}(q_{2i})-(-1)^kT^{\{\overrightarrow{M}\}}(q_{2i})\cdot\partial^{-1}\cdot T^{\{\overrightarrow{M}\}}(q_{1i})\big),
\end{align*}
where $\overrightarrow{M}\setminus\{j\}=(M,M-1,\cdots,j+1,j-1,\cdots,2,1)$.
\end{proposition}
\begin{proof}
When $M=0$, it holds obviously. Assume that this proposition holds for $M$, we then prove that it also holds for $M+1$.

First we know
$T^{\{\overrightarrow{M+1}\}}=T[f_{M+1}^{\{M\}}]T^{\{\overrightarrow{M}\}}$,
then by Lemma \ref{lemma:TLKT},
\begin{align*}
\bigl(T^{\{\overrightarrow{M+1}\}}L^kT^{\{\overrightarrow{M+1}\}-1}\bigr)_{<0}
&=T\bigl[f_{M+1}^{\{{M}\}}\bigr]L^{\{M\}k}\bigl(f_{M+1}^{\{{M}\}}\bigr)\cdot\partial^{-1}\cdot\bigl(T_i[f_{M+1}^{\{{M}\}}]\bigl(f_{M+1}^{\{{M}\}}\bigr)\bigr)^{-1}\\
&-(-1)^k\bigl(T_i[f_{M+1}^{\{{M}\}}]\bigl(f_{M+1}^{\{{M}\}}\bigr)\bigr)^{-1}\cdot\partial^{-1}\cdot T\bigl[f_{M+1}^{\{{M}\}}\bigr]L^{{\{m}\}k}\bigl(f_{M+1}^{\{{M}\}}\bigr)\\
&+\sum_{j=1}^M\big(T\bigl[f_{M+1}^{\{{M}\}}\bigr]T^{\{\overrightarrow{M}\}}L^k(f_j)\cdot\partial^{-1}\cdot T[f_{M+1}^{\{{M}\}}]\bigl(T^{[\overrightarrow{M}\setminus\{j\},\overrightarrow{M}]}(f_j)\bigr)^{-1}\big)\\
&-(-1)^kT[f_{M+1}^{\{{M}\}}]\bigl(T^{[\overrightarrow{M}\setminus\{j\},\overrightarrow{M}]}(f_j)\bigr)^{-1}\cdot\partial^{-1}\cdot T\bigl[f_{M+1}^{\{{M}\}}\bigr]T^{\{\overrightarrow{M}\}}L^k(f_j)+\\
&\sum_{i=1}^m\Bigl(T^{\{\overrightarrow{M+1}\}}(q_{1i})\cdot\partial^{-1}\cdot T^{\{\overrightarrow{M+1}\}}(q_{2i})-(-1)^kT^{\{\overrightarrow{M+1}\}}(q_{2i})\cdot\partial^{-1}\cdot T^{[\overrightarrow{M+1}]}(q_{1i})\Bigr)
\end{align*}
Note that
\begin{align*}
&T\bigl[f_{M+1}^{\{{M}\}}\bigr]L^{\{M\}k}\bigl(f_{M+1}^{\{{M}\}}\bigr)=\left(T^{\{\overrightarrow{M+1}\}}L^k\right)(f_{M+1}),\\ &T\bigl[f_{M+1}^{\{{M}\}}\bigr]T^{\{\overrightarrow{M}\}}L^k(f_j)=\left(T^{\{\overrightarrow{M+1}\}}L^k\right)(f_{j}),\\
&\bigl(T_i\bigl(f_{M+1}^{\{{M}\}}\bigr)\bigl(f_{M+1}^{\{M\}}\bigr)\bigr)^{-1}=\left(T^{[\overrightarrow{M+1}\setminus\{M+1\},\overrightarrow{M+1}]}(f_{M+1})\right)^{-1},
\end{align*}
and by Lemma \ref{lem:T}
\begin{align*}
T\bigl(f_{M+1}^{\{\overrightarrow{M}\}}\bigr)\bigl(T^{[\overrightarrow{M}\setminus\{j\},\overrightarrow{M}]}(f_j)\bigr)^{-1}
=\left(T^{[\overrightarrow{M+1}\setminus\{j\},\overrightarrow{M+1}]}(f_j)\right)^{-1},
\end{align*}
Substituting these four relations into $\bigl(T^{\{\overrightarrow{M+1}\}}L^kT^{\{\overrightarrow{M+1}\}-1}\bigr)_{<0}$, we can find this proposition  holds for $M+1$.
\end{proof}
\begin{corollary}\label{fi}
Given $f_{i}(1\leq i\leq m)$ satisfying $f_{i_{t_n}}=f_i^{(n)}$ and  $\varphi_j=e^{\xi(t,\la_j)}(1\leq j\leq P)$, if denote $T^{\{\overrightarrow{M}\}}=T^{\{\overrightarrow{M}\}}(f_1,f_1^{(k)},\cdots,f_1^{(kN_1)},\cdots,f_m,f_m^{(k)},\cdots,f_m^{(kN_m)},\varphi_1,\cdots,\varphi_P)$ with $M=
P+m+\sum_{j=1}^mN_j$, then
\begin{align}\label{TM}
(T^{\{\overrightarrow{M}\}}\cdot\pa^k\cdot T^{\{\overrightarrow{M}\}-1} )_{<0}
&=\sum_{j=1}^m\big(T^{\{\overrightarrow{M}\}}(f_j^{(k(N_j+1))})\cdot\pa^{-1}\cdot \left(T^{[\overrightarrow{M}\setminus\{p(j)\}\, ,\overrightarrow{M}]}(f_j^{(kN_j)})\right)^{-1}\nonumber\\
&-(-1)^k\big(T^{[\overrightarrow{M}\setminus\{p(j)\}\, ,\overrightarrow{M}]}(f_j^{(kN_j)})\big)^{-1}\cdot\pa^{-1}\cdot \big(T^{\{\overrightarrow{M}\}}(f_j^{(k(N_j+1)})\big),
\end{align}
where $p(j)=j+\sum_{l=1}^jN_j$.
\end{corollary}
\begin{proof}
First by Proposition \ref{Lax_(k,m)}, we have
\begin{align*}
(T^{\{\overrightarrow{M}\}}\cdot\pa^k\cdot T^{\{\overrightarrow{M}\}-1} )_{<0}
=\sum_{j=1}^m\sum_{l=0}^{N_j}T^{\{\overrightarrow{M}\}} (f_j^{((l+1)k)})\cdot\pa^{-1}\cdot \big(T^{[\overrightarrow{M}\setminus\{\sum_{l=1}^{j-1}N_l+j+l\}\, ,\overrightarrow{M}]}(f_j^{(lk)}) \big)^{-1}\\
-(-1)^k\sum_{j=1}^m\sum_{l=0}^{N_j}\big(T^{[\overrightarrow{M}\setminus\{\sum_{l=1}^{j-1}N_l+j+l\}\, ,\overrightarrow{M}]}(f_j^{(lk)}) \big)^{-1}\cdot\pa^{-1}\cdot T^{\{\overrightarrow{M}\}} (f_j^{((l+1)k)})\\
+\sum_{i=1}^P(T^{\{\overrightarrow{M}\}}\pa^k)(\psi_i)\cdot \pa^{-1}\cdot \big(T^{[\overrightarrow{M}\setminus\{\sum_{l=1}^mN_l+m+i\}\, ,\overrightarrow{M}]}(\varphi_i)\big)^{-1}\\
-(-1)^k\sum_{i=1}^P\big(T^{[\overrightarrow{M}\setminus\{\sum_{l=1}^mN_l+m+i\}\, ,\overrightarrow{M}]}(\varphi_i)\big)^{-1}\cdot\pa^{-1}\cdot (T^{\{\overrightarrow{M}\}}\pa^k)(\varphi_i).
\end{align*}
Then note that for $0\leq l \leq N_j-1$
\begin{align*}
T^{[\overrightarrow{M}]}(f_j^{((l+1)k)})=0,
\quad
T^{[\overrightarrow{M}]} (\varphi_i^{(k)})=\la_i^kT^{[\overrightarrow{M}]}(\varphi_i)=0,
\end{align*}
we can know \eqref{TM} holds.
\end{proof}
\subsection{Proof of {\bf Theorem} \ref{THE}}
After the above preparation, we give the proof of {\bf Theorem} \ref{THE}. Firstly by corollary \ref{fi}, if denote
\begin{align*}
&L^{\{M\}}=T^{\{\overrightarrow{M}\}}\cdot \pa \cdot T^{\{\overrightarrow{M}\}-1},\\
&q_{1,j}^{\{M\}}=T^{\{\overrightarrow{M}\}}(f_j^{(k(N_j+1)}),\\
&q_{2,j}^{\{M\}}=\big(T^{[\overrightarrow{M}\setminus\{j+\sum_{l=1}^jN_l\}\, ,\overrightarrow{M}]}(f_j^{(kN_j)})\big)^{-1},
\end{align*}
then we have
\begin{align}\label{LMK}
(L^{\{M\}k})_{<0}=\sum_{i=1}^{m}\big(q_{1,j}^{\{M\}}\partial^{-1}q_{2,j}^{\{M\}}-(-1)^kq_{2,j}^{\{M\}}\partial^{-1}q_{1,j}^{\{M\}}\big).
\end{align}

 Notice that $(L^{\{0\}},q_{1,i}^{\{0\}},q_{2,i}^{\{0\}})=(\pa ,0,0) $ is a trivial solution to the $(k,m)$-constrained KP system. Therefore, \eqref{LMK} is the transformed result of the above trivial solution under the CKP Darboux transformation $T^{\{\overrightarrow{M}\}}$, and is also a solution to the $(k,m)$-constrained CKP hierarchy, also satisfying (please refer to \eqref{tm})
\begin{align}
&L_{t_n}^{\{M\}}=[(L^{\{M\}n})_{\geq0},L^{\{M\}}], \quad L^{\{M\}*}=-L^{\{M\}},  \\
&\pa_{t_n}q_{a,i}^{\{M\}}=(L^{\{M\}n})_{\geq0}(q_{a,i}^{\{M\}}), \quad a=1,2, \quad i=1,2,\ldots,m.
\end{align}

Next, note that the tau-function $\tau^{\{0\}}=1$ corresponding to the CKP Lax operator $L^{\{0\}}=\pa$, then by \eqref{mm},
\begin{align*} \tau^{\{M\}}=\sqrt{IW_{M,M}(h_1,\ldots,h_M;h_1,\ldots,h_M)} \end{align*}
is the CKP tau-function. Therefore
\begin{align*}
\psi^{\{M\}}(t,\la)=e^{\xi(t,\la)}\sqrt{\varphi^{\{M\}}(t,\la)}\frac{\tau^{\{M\}}(t-2[\la^{-1}])}{\tau^{\{M\}}(t)}
\end{align*}
is the corresponding CKP wave function, satisfying
\begin{align}
		{\rm Res}_\lambda \psi^{\{M\}}(t,\lambda)\psi^{\{M\}}(t',-\lambda)=0.
	\end{align}
So by Theorem \ref{k,m}, we can know
\begin{align}
				&{\rm Res}_\la \la^k\psi^{\{M\}}(t,\la)\psi^{\{M\}}(t',-\la)=\sum^m_{i=1}(q_{1,i}^{\{M\}}(t)q_{2,i}^{\{M\}}(t')-(-1)^kq_{2,i}^{\{M\}}(t)q_{1,i}^{\{M\}}(t')), \label{theorem2:bili1}\\
				&{\rm Res}_\la \psi^{\{M\}}(t,\la)\Om^{\{M\}}(q_{a,i}(t'),\psi^{\{M\}}(t',-\la))=-q_{a,i}^{\{M\}}(t), \quad a=1,2, \quad i=1,2,\ldots,m.
				\label{theorem2:bili2}
			\end{align}
Further if we set
\begin{align*}
\rho_{1,i}^{\{M\}}=q_{1,i}^{\{M\}}\tau^{\{M\}}, \quad
\rho_{2,i}^{\{M\}}=q_{2,i}^{\{M\}}\tau^{\{M\}},
\end{align*}
then according to Theorem \ref{k,m} again, we know they satisfy \eqref{res1} and \eqref{res2}.

According to \eqref{2M=} \eqref{2M<}, we know
\begin{align*}
&q_{1,j}^{{\{M}\}}=T^{\{\overrightarrow{M}\}}(f_j^{(kN_j+k)})
=\frac{IW_{M,M+1}(h_1,\ldots,h_M;h_1,\ldots,h_M,f_j^{(kN_j+k)})}{IW_{M,M}(h_1,\ldots,h_M;h_1\ldots,h_M)},\\
&q_{2,j}^{{\{M}\}}=(T^{[\overrightarrow{M}\setminus\{j+\sum_{l=1}^jN_l\},\overrightarrow{M}]}(f_j^{(kN_j)}))^{-1}
=\frac{(-1)^{p(j)+M}IW_{M-1,M}(h_1,\ldots,\widehat{h_{p(j)}},\ldots,h_M;h_1.\ldots,h_M)}{IW_{M,M}(h_1,\ldots,h_M;h_1,\ldots,h_M
)},
\end{align*}
where $p(j)=j+\sum_{l=1}^jN_l$, implying
\begin{align*}
&\rho_{1,i}^{\{M\}}=\frac{IW_{M,M+1}(h_1,\ldots,h_M;h_1,\ldots,h_M,f_j^{(kN_j+k)})}{\sqrt{IW_{M,M}(h_1,\ldots,h_M;h_1,\ldots,h_M)}},\\
&\rho_{2,i}^{\{M\}}=\frac{(-1)^{M+p(j)}IW_{M-1,M}(h_1,\ldots,\widehat{h_{p(j)}},\ldots,h_M;h_1,\ldots,h_M)}{\sqrt{IW_{M,M}(h_1,\ldots,h_M;h_1,\ldots,h_M)}}.
\end{align*}
Therefore we have proved {\bf Theorem} \ref{THE}.

\section{Examples of the constrained CKP hierarchy}
In this section, we will list some explicit examples of the constrained CKP hierarchy and give the corresponding solutions.
\subsection{The $(1,m)$-constrained CKP hierarchy}
	When $k=1$, the Lax operator $L$ of the $(1,m)$-constrained CKP hierarchy is given by
\begin{align}\label{(1,m)}
	L=\pa+\sum_{i=1}^m\left(q_{1i}\pa^{-1}q_{2i}+q_{2i}\pa^{-1}q_{1i}\right).
\end{align}
\begin{example}
	$q_{1,j}$ and $q_{2,j}$ in \eqref{(1,m)} will satisfy the following equations
\begin{align*}
	&q_{1j,t_3}=q_{1j,xxx}+3\sum_{i=1}^m(2q_{1i}q_{2i}q_{1j,x}+q_{2i}q_{1i,x}q_{1j}+q_{1i}q_{2i,x}q_{1j}),\\
	&q_{2j,t_3}=q_{2j,xxx}+3\sum_{i=1}^m(2q_{1i}q_{2i}q_{2j,x}+q_{1i,x}q_{2i}q_{2j}+q_{1i}q_{2i,x}q_{2j}).
\end{align*}
In particular, when $m=1$, we set $q_1\triangleq q_{1,1}$, $q_2\triangleq q_{2,1}$, $t\triangleq t_3$ then we have
\begin{align}\label{q12}
	&q_{1,t}=q_{1,xxx}+9q_1q_2q_{1,x}+3q_1^2q_{2,x},\nonumber\\
	&q_{2,t}=q_{2,xxx}+9q_1q_2q_{2,x}+3q_2^2q_{1,x}.
\end{align}
When $q_1=q_2$, \eqref{q12} is the mKdV equation
\begin{align*}
q_{t}=q_{xxx}+9q^2q_{x}+3q^2q_x.
\end{align*}

Next we will give the solutions of \eqref{q12}. For this, let us set $N_1=1$, $P=0$, $M=2$, $m=1$,  $f_1=e^{x+t}+e^{2x+8t}+e^{3x+27t}$ in Theorem \ref{THE} then we can know
\begin{align*}
	\tau^{\{2\}}=& \sqrt{
		\begin{vmatrix}
			\Omega(f_1,f_1) & \Omega(f_{1,x},f_1) \\
			\Omega(f_1,f_{1,x}) & \Omega(f_{1,x},f_{1,x})
		\end{vmatrix}
	}\\
	=&(\frac{1}{600}e^{10x+70t}+\frac{1}{72}e^{6x+18t}+\frac{1}{12}e^{8x+56t}\\
	+&\frac{1}{15}e^{7x+37t}+\frac{1}{120}e^{8x+44t}+\frac{1}{45}e^{9x+63t})^{1/2},\\
\rho_{1}^{\{2\}} =&
	\frac{
		\begin{vmatrix}
			\Omega(f_1,f_1) & \Omega(f_{1,x},f_1) & \Omega(f_{1,xx},f_1) \\
			\Omega(f_1,f_{1,x}) & \Omega(f_{1,x},f_{1,x}) & \Omega(f_{1,xx},f_{1,x}) \\
			f_1 & f_{1,x} & f_{1,xx}
		\end{vmatrix}
	}{
		\sqrt{
			\begin{vmatrix}
				\Omega(f_1,f_1) & \Omega(f_{1,x},f_1) \\
				\Omega(f_1,f_{1,x}) & \Omega(f_{1,x},f_{1,x})
			\end{vmatrix}
		}
	}\\
	=&\frac{1}{30}(10e^{10x+64t}+5e^{9x+45t}+e^{11x+71t})/(6e^{10x+71t}+50e^{6x+18t}\\
+&300e^{8x+56t}+240e^{7x+37t}+30e^{8x+44t}+80e^{9x+63t})^{1/2},
\end{align*}
\begin{align*}	
	\rho_{2}^{\{2\}}=&\frac{
		\begin{vmatrix}
			\Omega(f_1,f_1)  & \Omega(f_{1,x},f_1)\\
			f_1 & f_{1,x}
		\end{vmatrix}
	}{
		\sqrt{
			\begin{vmatrix}
				\Omega(f_1,f_1) & \Omega(f_{1,x},f_1) \\
				\Omega(f_1,f_{1,x}) & \Omega(f_{1,x},f_{1,x})
			\end{vmatrix}
		}
	}\\
	=&(10e^{4x+10t}+30e^{5x+29t}+5e^{5x+17t}+24e^{6x+36t}+10e^{7x+55t}\\
	+&3e^{7x+43t}+2e^{8x+62t})/(6e^{10x+70t}+50e^{6x+18t}\\
	+&300e^{8x+56t}+240e^{7x+37t}+30e^{8x+44t}+80e^{9x+63t})^{1/2},
\end{align*}
	where we have set $\rho_a=\rho_{a,i}$.
	 \begin{align*}
	 	&\Omega(f_1,f_1)=\frac{1}{2}e^{2x+2t}+\frac{2}{3}e^{3x+9t}+\frac{1}{2}e^{4x+28t}+\frac{1}{4}e^{4x+16t}+\frac{2}{5}e^{5x+35t}+\frac{1}{6}e^{6x+54t},\\
	 	&\Omega(f_{1,x},f_1)=\frac{1}{2}e^{2x+2t}+e^{3x+9t}+e^{4x+28t}+\frac{1}{2}e^{4x+16t}+e^{5x+35t}+\frac{1}{2}e^{6x+54t},\\
	 	&\Omega(f_{1,xx},f_1)=\frac{1}{2}e^{2x+2t}+e^{4x+16t}+\frac{3}{2}e^{6x+54t}+\frac{5}{3}e^{3x+9t}+\frac{5}{2}e^{4x+28t}+\frac{13}{5}e^{5x+35t},\\
	 	&\Omega(f_{1,xx},f_{1,x})=\frac{1}{2}e^{2x+2t}+2e^{3x+9t}+3e^{4x+28t}+2e^{4x+16t}+6e^{5x+35t}+\frac{9}{2}e^{6x+54t},\\
	 	&\Omega(f_{1,x},f_{1,x})=\frac{1}{2}e^{2x+2t}+\frac{4}{3}e^{3x+9t}+\frac{3}{2}e^{4x+28t}+e^{4x+16t}+\frac{12}{5}e^{5x+35t}+\frac{3}{2}e^{6x+54t}.
	 \end{align*}
	 Then by Theorem  \ref{THE} we can know
\begin{align*}
	q_{1}=\frac{\rho_{1}^{\{2\}}}{\tau^{\{2\}}}=&(10e^{4x+46t}+5e^{3x+27t}+e^{5x+53t})/(25+120e^{x+19t}\\
	&+150e^{2x+38t}+15e^{2x+26t}+40e^{3x+45t}+3e^{4x+52t}),\\
	q_{2}=\frac{\rho_{2}^{\{2\}}}{\tau^{\{2\}}}=&(30(10+30e^{x+19t}+5e^{x+7t}+24e^{2x+26t}\\
	&+10e^{3x+45t}+3e^{3x+33t}+2e^{4x+52t}))/(25e^{2x+8t}+120e^{3x+27t}\\
	&+150e^{4x+46t}+15e^{4x+34t}+40e^{5x+53t}+3e^{6x+60t}).
\end{align*}
\end{example}
\begin{example}
In particular when $m=1$, we set $q_{1,1}\triangleq q_1$, $q_{2,1}\triangleq q_2$, $t\triangleq t_5$, then we have
\begin{align}\label{x^5}
	q_{1,t}=&q_{1,xxxxx}+20q_{1,x}^2q_{2,x}+30q_2q_{1,x}q_{1,xx}+25q_1q_{2,x}q_{1,xx}\nonumber\\
	         \ &+25q_1q_{1,x}q_{2,xx}+15q_1q_2q_{1,xxx}+5q_1^2q_{2,xxx}+80q_1^2q_2^2q_{1,x}+40q_1^3q_2q_{2,x},\nonumber\\
	q_{2,t}=&q_{2,xxxxx}+20q_{2,x}^2q_{1,x}+30q_1q_{2,x}q_{2,xx}+25q_2q_{1,x}q_{2,xx}\nonumber\\
	         \ &+25q_2q_{2,x}q_{1,xx}+15q_2q_1q_{2,xxx}+5q_2^2q_{1,xxx}+80q_2^2q_1^2q_{2,x}+40q_2^3q_1q_{1,x}.
\end{align}
Next we will give the solutions of \eqref{x^5}. For this, let us set $N_1=1$, $P=0$, $M=2$, $m=1$,  $f_1=x^2$,  in Theorem \ref{THE}, then
\begin{align*}
	\tau^{\{2\}} =& \sqrt{
		\begin{vmatrix}
			\Omega(f_1,f_1) & \Omega(f_{1,x},f_1) \\
			\Omega(f_1,f_{1,x}) & \Omega(f_{1,x},f_{1,x})
		\end{vmatrix}
	}
	=\sqrt{(\frac{1}{60}x^8+\frac{16}{3}x^3t)},\\
	\rho_{1}^{\{2\}} =&
	\frac{
		\begin{vmatrix}
			\Omega(f_1,f_1) & \Omega(f_{1,x},f_1) & \Omega(f_{1,xx},f_1) \\
			\Omega(f_1,f_{1,x}) & \Omega(f_{1,x},f_{1,x}) & \Omega(f_{1,xx},f_{1,x}) \\
			f_1 & f_{1,x} & f_{1,xx}
		\end{vmatrix}
	}{
		\sqrt{
			\begin{vmatrix}
				\Omega(f_1,f_1) & \Omega(f_{1,x},f_1) \\
				\Omega(f_1,f_{1,x}) & \Omega(f_{1,x},f_{1,x})
			\end{vmatrix}
		}
	}
	=\frac{\sqrt{15}}{45}\frac{x^3(x^5-480t)}{\sqrt{x^3(x^5+320t)}},\\
	\rho_{2}^{\{2\}}=&\frac{
		\begin{vmatrix}
			\Omega(f_1,f_1)  & \Omega(f_{1,x},f_1)\\
			f_1 & f_{1,x}
		\end{vmatrix}
	}{
		\sqrt{
			\begin{vmatrix}
				\Omega(f_1,f_1) & \Omega(f_{1,x},f_1) \\
				\Omega(f_1,f_{1,x}) & \Omega(f_{1,x},f_{1,x})
			\end{vmatrix}
		}
	}
	=\frac{\sqrt{15}}{5}\frac{x(-x^5+80t)}{\sqrt{x^3(x^5+320t)}}.
\end{align*}
Here $\Omega$ is computed by the way below
\begin{align*}
	\Omega(f,g)=\int_0^1 \left(xA(xy,ty)+tB(xy,ty)\right) dy ,
\end{align*}
where $A=\Omega(f,g)_x=fg$, $B=\Omega(f,g)_t=f_{xxxx}g-f_{xxx}g_{x}+f_{xx}g_{xx}-f_{x}g_{xxx}+fg_{xxxx}.$ It can be found that
\begin{align*}
	&\Omega(f_1,f_1)=\frac{1}{5}x^5+4t, \
	\Omega(f_{1,x},f_1)=\frac{1}{2}x^4, \
	\Omega(f_{1,xx},f_1)=\frac{2}{3}x^3,\\
	&\Omega(f_{1,xx},f_{1,x})=2x^2, \
	\Omega(f_{1,x},f_{1,x})=\frac{4}{3}x^3,
\end{align*}
then we have
\begin{align*}
	q_1=\frac{2}{3}\frac{x^5-480t}{x^5+320t}, \quad
	q_2=\frac{6(-x^5+80t)}{x^2(x^5+320t)}.
\end{align*}
\end{example}
\subsection{The $(2,1)$-constrained CKP hierarchy}
When $k=2, m=1$, the Lax operator $L$ of the constrained CKP hierarchy is given by
\begin{align*}
	L^2=\pa^2+u+q_{1}\pa^{-1}q_{2}-q_{2}\pa^{-1}q_{1},
\end{align*}
In terms of tau functions, we have $(1\leq j \leq m)$
\begin{align*}
	&u=4\frac{\tau(t)_{xx}\tau(t)-\tau(t)_x^2}{\tau^2(t)},\\
	&q_1=\frac{\rho_{1}}{\tau}, \quad q_2=\frac{\rho_{2}}{\tau}.
\end{align*}
	\begin{example}
		$u$, $q_1$, $q_2$ satisfy the following equations
\begin{align}
	&u_{t}=3q_{2}q_{1,xx}+\frac{3}{2}uu_x-3q_{1}q_{2,xx}+\frac{1}{4}u_{xxx},\label{coupledkdv1}\\
	&q_{1,t}=q_{1,xxx}+\frac{3}{2}uq_{1,x} +\frac{3}{4}u_xq_1,\label{coupledkdv2}\\
	&q_{2,t}=q_{2,xxx}+\frac{3}{2}uq_{2,x} +\frac{3}{4}u_xq_2.\label{coupledkdv3}
\end{align}
\eqref{coupledkdv1}-\eqref{coupledkdv3} is called the coupled KdV system \cite{loris2001jpa},
which is the generalization of the KdV equation.

Next we will give the solutions of \eqref{coupledkdv1}-\eqref{coupledkdv3}. For this, let us set $N_1=0$, $P=1$, $M=2$, $m=1$,  $f=x^2$, $\varphi=e^{x+t}$ in Theorem \ref{THE}, then we have
\begin{align*}
q_{1} =&
\frac{
	\begin{vmatrix}
		\Omega(f,f) & \Omega(\varphi,f) & \Omega(f^{(2)},f) \\
		\Omega(f,\varphi) & \Omega(\varphi,\varphi) & \Omega(f^{(2)},\varphi) \\
		f & \varphi & f^{(2)}
	\end{vmatrix}
}{
		\begin{vmatrix}
			\Omega(f,f) & \Omega(\varphi,f) \\
			\Omega(f,\varphi) & \Omega(\varphi,\varphi)
		\end{vmatrix}	
}\\
=&\frac{4(x^5-10x^4+15tx^2+40x^3-60tx-90x^2+60t+120x-60)}{3(x^5-10x^4+40x^3-80x^2+80x-40)},\\
q_{2} =&-
\frac{
	\begin{vmatrix}
		\Omega(f,\varphi)  & \Omega(\varphi,\varphi)\\
		f & \varphi
	\end{vmatrix}
}{
	\begin{vmatrix}
		\Omega(f,f) & \Omega(\varphi,f) \\
		\Omega(f,\varphi) & \Omega(\varphi,\varphi)
	\end{vmatrix}	
}\\
=&-\frac{5(x^2-4x+4)}{x^5-10x^4+40x^3-80x^2+80x-40},\\ u=&-\frac{10x(x^7-16x^6+112x^5-448x^4+1120x^3-1760x^2+1600x-640)}{(x^5-10x^4+40x^3-80x^2+80x-40)^2}.
\end{align*}
\end{example}

\section{Conclusions and Discussions}
In this paper, we construct the solutions of the $(k,m)$-constrained CKP hierarchy \eqref{Lax op}-\eqref{LAx} in terms of CKP tau functions (see \eqref{bili*}) by the KP Darboux transformations. We firstly convert the $(k,m)$-constrained CKP hierarchy \eqref{Lax op}-\eqref{LAx} into the equivalent bilinear equations \eqref{theorem:bili1}-\eqref{theorem:bili2} in terms of wave functions, then these bilinear equations are rewritten by introducing the CKP tau function $\tau(t)$ and auxiliary functions $\rho_{a,i}(t)$ (see \eqref{res1}-\eqref{res2}). After that, the explicit solutions for the CKP tau function $\tau(t)$ and auxiliary functions $\rho_{a,i}(t)$ are obtained in Theorem \ref{THE}, satisfying the above equivalent bilinear equations of the constrained CKP hierarchy in terms of wave functions.

Compared with the existed results on the constrained CKP hierarchy, our results have the following improvements:
\begin{itemize}
	\item The tau functions are used in the discussions of the constrained CKP hierarchy. Notice that the CKP tau function $\tau(t)$ in \eqref{bili*} is quite special due to the square root term in \eqref{bili*}, but it contains all the CKP information. Here we establish the relation for the CKP tau function $\tau(t)$ that can determine the whole $(k,m)$-constrained CKP hierarchy \eqref{Lax op}-\eqref{LAx}.
	\item The $(k,m)$-constrained CKP hierarchy with even $k$ is discussed. As far as we can know, there are very few papers for even $k$. \cite{loris2001jpa} is the only reference that we can know, where the system for even $k$ is just raised without discussing the corresponding solutions. The $(k,m)$-constrained CKP hierarchy with even $k$ is not the classical symmetry reduction for the CKP hierarchy.
	\item Our solutions are comparatively more general. For the $(k,1)$-constrained CKP hierarchy with odd $k$, the corresponding solutions are given in \cite{Loris-1999}, and the determinant formulas are given in \cite{he2007jmp} without explicit solutions. Notice that the cases of $m>1$ are missing. In fact, it is usually very difficult to determine the seed solutions for $m>1$. For $m=1$, we can consider the initial solution $L^k=\partial^k+\partial^{-1}$, and it is not easy to fix the explicit expression of eigenfunction $q$ which satisfies $q_{t_n}=(L^n)_{\geq 0}(q)$. While in the case of $m>1$, it becomes more complicated in determining the seed solutions. While in Eq, the condition is not easy to be generalized into the case of $m>1$. But here, we overcome these questions by starting from the initial solution $L=\partial$, where one can easily obtain the eigenfunction $q$ satsifying $q_{t_n}=q^{(n)}$. And our results contain the ones in \cite{Loris-1999} by setting $m=1$ and $P=0$ in Theorem \ref{THE}.
\end{itemize}
\noindent{\bf Acknowledgements}: \\
This article is dedicated to Professor Ke Wu in Capital Normal University in celebration of his 80th birthday. And this work is supported by National Natural Science Foundation of China
(Grant Nos. 12571271 and 12261072).\\

\noindent{\bf Conflict of Interest}: \\
The author have no conflicts to disclose.\\

\noindent{\bf Data availability}: \\
Date sharing is not applicable to this article as no new data were created or analyzed in this study.\\


\begin{thebibliography}{99}
\bibitem{Arthamonov-2023}
Arthamonov S, Harnad J and Hurtubise J,
Lagrangian Grassmannians, CKP hierarchy and hyperdeterminantal relations,
Comm. Math. Phys. 401 (2023) 1337-1381.


\bibitem{chang2013}
Chang L and Wu C Z,
Tau function of the CKP hierarchy and nonlinearizable Virasoro symmetries,
Nonlinearity 26 (2013) 2577-2596.

\bibitem{cheng2014}
Cheng J P and He J S,
The ``ghost" symmetry in the CKP hierarchy,
J. Geom. Phys. 80 (2014) 49-57.



\bibitem{ChengY-1992}
Cheng Y,
Constraints of the Kadomtsev-Petviashvili hierarchy,
J. Math. Phys. 33 (1992) 3774-3782.


\bibitem{ChengY-1994}
Cheng Y and Zhang Y J,
Bilinear equations for the constrained KP hierarchy,
Inverse Problems 10 (1994) L11-L17.

\bibitem{date1981JPAJ}
Date E, Jimbo M, Kashiwara M and Miwa T,
Transformation groups for soliton equations. VI. KP hierarchies of orthogonal and symplectic type,
J. Phys. Soc. Japan  50 (1981) 3813-3818.

\bibitem{date1983}
Date E, Jimbo M, Kashiwara M and MiwaT,
Transformation groups for soliton equations, in {\it Nonlinear integrable systems-classical theory and quantum theory\/}, ed. by M. Jimbo and T. Miwa, World Scientific, Singapore, 1983, pp. 39-119.

\bibitem{dickey2003}
Dickey L A,
Soliton Equations and Hamiltonian Systems, in \emph{Advanced Series in Mathematical Physics}, 26. World Scientific Publishing Co., Inc., River Edge, NJ, 2003.



\bibitem{he2002}
He J S, Li Y S and Cheng Y, The determinant representation of the gauge transformation operators, Chin. Ann. Math. Ser. B. 23 (2002) 475-486.

\bibitem{he2007jmp}
He J S, Wu Z W and Cheng Y,
Gauge transformations for the constrained CKP and BKP hierarchies,
J. Math. Phys. 48 (2007) 113519.




\bibitem{Konopelchenko1984}
Konopelchenko B G and Dubrovsky V G, Some new integrable nonlinear
evolution equations in 2+1 dimensions, Phys. Lett. A. 102 (1984) 15-17.

\bibitem{Krichever-CMP-2021}
Krichever I and Zabrodin A,
Kadomtsev-Petviashvili turning points and CKP hierarchy,
Comm. Math. Phys. 386 (2021) 1643-1683.

\bibitem{Krichever-Zabrodin-LMP-2021}
Krichever I and Zabrodin A V,
Constrained Toda hierarchy and turning points of the Ruijsenaars-Schneider model,
Lett. Math. Phys. 112 (2022) Paper No. 23.

\bibitem{lichunxia}
 Li C X and Li S H, The Cauchy two-matrix model, C-Toda lattice and CKP hierarchy, J. Nonlinear Sci. 29 (2019) 3-27.

\bibitem{lishihao}Li S H,  Tsujimoto S, Watanabe R and  Yu G F,Cauchy-Jacobi orthogonal polynomials and the discrete CKP equation, Lett. Math. Phys. 115 (2025) 116.

\bibitem{Loris-1999}
Loris I,
On reduced CKP equations,
Inverse Problems 15 (1999) 1099-1109.

\bibitem{loris2001jpa}
 Loris I,
 Dimensional reductions of BKP and CKP hierarchies,
 J. Phys. A. 34 (2001) 3447-3459.

\bibitem{oevel1993pa}
Oevel W, Darboux theorems and Wronskian formulas for integrable systems. I. Constrained KP flows, Phys. A 195 (1993) 533-576.



\bibitem{Sidorenko-1991}
Sidorenko J and Strampp W,
Symmetry constraints of the KP hierarchy,
Inverse Problems 7 (1991) L37-L43.

\bibitem{vandeleur2012}
van de Leur J, Orlov A Y and Shiota T,
CKP hierarchy, bosonic tau function and bosonization formulae,
SIGMA Symmetry Integrability Geom. Methods Appl. 8 (2012) 036.

\bibitem{WangS-Non-2025}
Wang S, Guan W C and Cheng J P,
Bosonic construction of CKP tau function,
Nonlinearity 38 (2025) 015009.

\bibitem{WangS-PhD-2026}
Wang S, Wang J B, Guan W C and Cheng J P,
On symmetries of modified CKP hierarchy.
Phys. D 490 (2026) 135182.

\bibitem{yang2021}
Yang Y, Geng L M and Cheng J P,
CKP hierarchy and free bosons,
J. Math. Phys. 62 (2021) 083506.

\bibitem{yang2022}
Yang Y and Cheng J P, Bilinear equations in Darboux transformations by Boson--Fermion correspondence, Phys. D 433 (2022) 30.

\bibitem{Zabrodin-2023}
Zabrodin A V,
Tau-function of the multi-component CKP hierarchy,
Math. Phys. Anal. Geom. 27 (2024) Paper No.1.

\bibitem{Mulase1994}
Mulase M, Algebraic theory of the KP equations, in {\it Perspectives in mathematical physics}, Int. Press, Cambridge, MA, 1994, pp.151-217.

\bibitem{Zabrodin2021}
Zabrodin A V, Kadomtsev-Petviashvili hierarchies of types B and C, Theoret. Math. Phys. 208 (2021) 865-885.

\bibitem{Wangshen2025}
Wang S, Guan W C and Cheng J P, Tau functions of modified CKP hierarchy, J. Geom. Phys. 207 (2025) 105367.

\bibitem{Kacvandeleur2023}
Kac V and van de Leur J, Multicomponent KP type hierarchies and their reductions, associated to conjugacy classes of Weyl groups of classical Lie algebras, J. Math. Phys. 64 (2023) 091702.

\bibitem{OeveleStrampp}Oevel W and Strampp W, Constrained KP hierarchy and bi-Hamiltonian structures, Comm. Math. Phys. 157 (1993) 51-81.

\bibitem{Tiankl2011}Tian K L, He J S, Cheng J P and Cheng Y, Additional symmetries of constrained CKP and BKP hierarchies, Sci. China Math. 54 (2011) 257-268.


\bibitem{Cheng-Milanov} Cheng J P and Milanov T,  The extended D-Toda hierarchy. Selecta Math. (N.S.) 27 (2021) 24.


\bibitem{Wuyq2023}Wu Y Q and Cheng J P, A new generalized constrained modified KP hierarchy, Math. Methods Appl. Sci. 46 (2023) 3510-3521.

\bibitem{ChauShawTu1997}Chau L L, Shaw J C and Tu M H, Solving the constrained KP hierarchy by gauge transformations, J. Math. Phys. 38 (1997) 4128-4137.


\end{thebibliography}
\end{document}